\documentclass[11pt,reqno]{amsart}
\newtheorem{lemma}{Lemma}
\newtheorem{theorem}{THEOREM}

\newcommand{\R}{\mathbb{R}}
\renewcommand{\S}{\mathbb{S}}

\newcommand{\C}{\mathbb{C}}

\newcommand{\F}{\mathcal{F}}
\newcommand{\W}{\mathcal{W}}

\newcommand{\gm}{\gamma}
\newcommand{\al}{\alpha}

\newcommand{\half}{\mbox{$\frac 12$}}
\newcommand{\bra}{\langle}
\newcommand{\ket}{\rangle}

\newcommand{\be}{\begin{equation}}
\newcommand{\ee}{\end{equation}}
\newcommand{\bea}{\begin{align}}
\newcommand{\eea}{\end{align}}

\newcommand\de{\mathcal D}
\newcommand\infspec{{\rm{inf\, spec\,}}}
\newcommand\eps\epsilon
\newcommand\V{\mathcal{V}}
\newcommand\B{\mathcal{B}}
\newcommand\dig{\mathfrak{F}}

\DeclareMathOperator{\sgn}{sgn}

\DeclareMathOperator{\const}{const}

\begin{document}

\title
{Spectral properties of the
BCS gap equation of superfluidity}

\thanks{Plenary talk given by C. Hainzl at QMath10, $10^{\rm th}$ Quantum Mathematics
  International Conference, Moeciu, Romania, September 10--15, 2007.
  \\ \indent \copyright\, 2008 by the authors. This work may be reproduced, in
  its entirety, for non-commercial purposes.}

\author{Christian Hainzl} \address{Christian Hainzl, Departments of
  Mathematics and Physics, UAB, 1300 University Blvd,\\ Birmingham AL
  35294, USA}
 \email{hainzl@math.uab.edu}

\author{Robert Seiringer} \address{Robert Seiringer, Department of
  Physics, Princeton University, Princeton NJ 08542-0708, USA}
\email{rseiring@princeton.edu}

\date{Feb. 4, 2008}

\begin{abstract}
  We present a review of recent work on the mathematical aspects of
  the BCS gap equation, covering our results of \cite{HS} as well
  our recent joint work with Hamza and Solovej \cite{HHSS} and with
  Frank and Naboko \cite{FHNS}, respectively. In addition, we mention
  some related new results.
\end{abstract}

\maketitle

\section{Introduction}

In this paper we shall describe our recent mathematical study
 \cite{HHSS,FHNS,HS} of one of the current
hot topics in condensed matter physics, namely ultra cold fermionic
gases consisting of neutral spin-$\frac 12$ atoms. The kinetic energy
of these atoms is described by the non-relativistic Schr\"odinger
operator, and their interaction by a pair potential $\lambda V$ with
$\lambda$ being a coupling parameter. As experimentalists are nowadays
able to vary the inter-atomic potentials, the form of $\lambda V$ in
actual physical systems can be quite general; see the recent reviews
in \cite{Chen} and \cite{zwerger}. Our primary goal
concerns the study of the {\em superfluid phases} of such
systems. According to Bardeen, Cooper and Schrieffer \cite{BCS} (BCS)
the superfluid state is characterized by the existence of a
non-trivial solution of the {\em gap equation}
\begin{equation}\label{bcseintro}
\Delta(p) = -\frac \lambda{(2\pi)^{3/2}} \int_{\R^3} \hat V(p-q)
\frac{\Delta(q)}{E(q)} \tanh \frac{E(q)}{2T} \, dq
\end{equation}
at some temperature $T\geq 0$, with $E(p)= \sqrt{(p^2-\mu)^2 +
  |\Delta(p)|^2}$. Here, $\mu > 0$ is the chemical potential and $\hat
V(p)=(2\pi)^{-3/2}\int_{\R^3} V(x) e^{-ipx} dx $ denotes the Fourier
transform of $V$. The function $\Delta(p)$ is the order parameter and
represents the wavefunction of the {\em Cooper pairs}. Despite the
fact that the BCS equation \eqref{bcseintro} is highly non-linear, we
shall show in Theorem \ref{mthmbcs} (see also \cite[Thm 1]{HHSS}) that
the existence of a non-trivial solution to \eqref{bcseintro} at some
temperature $T$ is equivalent to the fact that a certain {\em linear
  operator}, given in \eqref{linopfinT} below, has a negative
eigenvalue. For $T=0$ this operator is given by $|-\Delta - \mu| +
\lambda V$. This rather astonishing possibility of reducing a
non-linear to a linear problem allows for a more thorough mathematical
study. Using spectral-theoretic methods, we are able to give a precise
characterization of the class of potentials leading to a non-trivial
solution for \eqref{bcseintro}.  In particular, in Theorem
\ref{thm2.2} (see also \cite[Thm 1]{FHNS}) we prove that for all
interaction potentials that create a negative eigenvalue of the
effective potential on the Fermi sphere (see \eqref{defvm} below; a
sufficient condition for this property is that $\int_{\R^3} V(x) dx <
0$), there exists a {\it critical temperature} $T_c(\lambda V) > 0$
such that \eqref{bcseintro} has a non-trivial (i.e., not identically
vanishing) solution for all $T < T_c(\lambda V)$, whereas there is no such 
solution for $T \geq T_c(\lambda V)$. Additionally, we shall determine
in Theorem~\ref{thm2.2} the precise asymptotic behavior of
$T_c(\lambda V)$ in the small coupling limit. We extend this result in
Theorem \ref{constant} (see also \cite[Thm 1]{HS}) and give a
derivation of the critical temperature $T_c$ valid to {\it second order} 
Born approximation. More precisely, we shall show that
\begin{equation}\label{form:tc}
T_c = \mu \frac{8 e^{\gamma-2}}{\pi} e^{\pi/(2 \sqrt{\mu} b_\mu)}
\end{equation}
where $\gamma\approx 0.577$ denotes
Euler's constant, and where $b_\mu<0$ is an effective scattering length. To
first order in the Born approximation, $b_\mu$ is related to the
scattering amplitude of particles with momenta on the Fermi sphere,
but to second order the expression is more complicated. The precise
formula is given in Eq.~(\ref{bmexpl}) below.  For interaction
potentials that decay fast enough at large distances, we shall show
that $b_\mu$ reduces to the usual {\it scattering length} $a_0$ of the
interaction potential in the low density limit, i.e.,
for small $\mu$. Our formula thus represents a {\em generalization} of a well-known formula in the physics literature \cite{gorkov,NS}.\\

In the case of zero temperature, the function $E(p)$ in
\eqref{bcseintro} describes an effective energy-momentum relation for
quasi particles, and $$\Xi : = \inf_{p} E(p) = \inf_p \sqrt{(p^2 -
  \mu)^2 + |\Delta(p)|^2}$$ is called the {\em energy gap} of the
system. It is of major importance for applications, such as the
classification of different types of superfluids. In fact, $\Xi$ is the
spectral gap of the corresponding second quantized BCS Hamiltonian. (See
\cite{BCS} and \cite{MR} or the appendix in \cite{HHSS}.)

An important problem is the classification of potentials $V$ for which
$\Xi > 0$. This questions turns out to be intimately related to the
continuity of the momentum distribution $\gamma(p)$, which will be
introduced in the next section.  In the normal (i.e., not superfluid)
state, $\Delta = 0$ and $\gamma$ is a step function at $T=0$, namely
$\gamma(p) = \theta(|p| - \sqrt{\mu})$. According to the picture
presented in standard textbooks the appearance of a superfluid phase
softens this step function and $\gamma(p)$ becomes continuous. We are
going to prove in this paper that if $ V(x)|x| \in L^{6/5}$ and $\int
V < 0$ then indeed both strict positivity of $\Xi > 0$ and continuity
$\gamma$ hold.  It remains an open problem to find examples of
potentials such that the gap vanishes in cases where a superfluid
phase occurs.

One of the difficulties involved in evaluating $\Xi$ is the potential
non-uniqueness of the solution of the BCS gap equation.  For
interaction potentials that have nonpositive Fourier transform, however, we
shall show that the BCS pair wavefunction {\it is} unique, and has
zero angular momentum. In this case, we shall prove in Theorem \ref{gap}
(see also \cite[Thm. 2]{HS}) similar results for $\Xi$ as for the
critical temperature.  It turns out that, at least up to second order
Born approximation,
\begin{equation}\label{form:xi}
\Xi = T_c \frac \pi{e^\gamma}
\end{equation}
in this case. This equality is valid for any density, i.e., for any
value of the chemical potential $\mu$. In particular, $\Xi$ has
exactly the same exponential dependence on the interaction
potential, described by $b_\mu$, as the critical temperature $T_c$.

\section{Preliminaries and main results}

We consider a gas of spin $1/2$ fermions at temperature $T\geq 0$
and chemical potential $\mu > 0$, interacting via a local
two-body interaction potential of the form $2\lambda V(x)$. Here,
$\lambda>0$ is a coupling parameter, and the factor $2$ is
introduced for convenience. We assume that $V$ is real-valued and
has some mild regularity properties, namely $V\in L^1(\R^3)\cap
L^{3/2}(\R^3)$. In the BCS approximation, the system is described by
the {\it BCS functional} $\F_T$, derived by Leggett in his seminal
paper \cite{Leg}, based on the original work of BCS \cite{BCS}. The BCS functional
$\F_T$ is related to the {\em pressure} of the system and is given
by
\begin{equation}\label{freeenergy}
\F_T(\gm,\al)= \int (p^2-\mu)\gm(p)dp+\int |\alpha(x)|^2 V(x)dx -T
S(\gm,\al),
\end{equation}
where the entropy $S$ is $$ S(\gamma,\alpha) = - \int {\rm
Tr}_{\C^2}\left[\Gamma(p)
  \log \Gamma(p)\right]dp, \qquad \Gamma(p) = \mbox{$\left(\begin{matrix}
\gamma(p) & \hat \alpha(p) \\ \overline{\hat \alpha(p)} & 1- \gamma(p)\\
\end{matrix}\right)$}.$$  The functions $\gamma(p)$ and $\hat\alpha(p)$ are interpreted as the momentum distribution and the Cooper pair wave function, respectively. The satisfy the matrix constraint $0\leq \Gamma(p)\leq 1$ for all $p\in \R^3$. 
In terms of the BCS functional the occurrence of superfluidity is
described by minimizers with $\alpha \neq 0$. We remark that in the
case of the Hubbard-model this functional was studied in \cite{BLS}.

For an arbitrary temperature $0 \leq T < \infty$ the BCS
gap equation, which is the Euler-Lagrange equation associated with
the functional $\F_T$, reads
\begin{equation}\label{bcse}
\Delta(p) = -\frac \lambda{(2\pi)^{3/2}} \int_{\R^3} \hat V(p-q)
\frac{\Delta(q)}{E(q)} \tanh \frac{E(q)}{2T} \, dq,
\end{equation}
where $E(p)= \sqrt{(p^2-\mu)^2 + |\Delta(p)|^2}$. The order
parameter $\Delta$ is related to the expectation value of the Cooper
pairs $\alpha$ via $2\alpha(p)= \Delta(p)/E(p)$. We present in the
following a thorough mathematical study of this equation. In order to do so, we shall not attack the equation \eqref{bcse}
directly, but exploit the fact that $\alpha$ is a critical point of
the semi-bounded functional $\F_T$.

The key to our studies is the observation in \cite{HHSS} that the
existence of a non-trivial solution to the {\em non-linear} equation
\eqref{bcse} can be reduced to a {\em linear} criterion, which can
be formulated as follows.

 \begin{theorem}[{\bf \cite[Theorem 1]{HHSS}}] \label{mthmbcs}
 Let $V \in L^{3/2}$, $\mu \in \R$, and $\infty > T \geq 0$. Define
 $$K_{T,\mu} = (p^2 -
\mu)\frac{e^{(p^2-\mu)/T}+1}{e^{(p^2-\mu)/T}-1}\, .$$
 Then
 the non-linear BCS equation \eqref{bcse} has a non-trivial solution
 if  and only if the
 linear operator
\begin{equation}\label{linopfinT}
K_{T,\mu} +  \lambda V \,,
\end{equation}
acting on $L^2(\R^3)$, has at least one negative eigenvalue.
\end{theorem}

Hence we are able to relate a non-linear problem to a linear problem
which is much easier to handle. The operator $K_{T,\mu}$ is understood as a
multiplication operator in momentum space. In the limit $T \to 0$
this operator reduces to $ |-\Delta - \mu| + \lambda V$.

\subsection{The critical temperature}

Theorem \ref{mthmbcs} enables a precise definition of the critical
temperature, by
\begin{equation}\label{crittemp}
T_c (\lambda V): = \inf \{ T \, | K_{T,\mu} + \lambda V \geq 0\}.
\end{equation}
The symbol $K_{T,\mu}(p)$ is point-wise monotone in $T$. This
implies that for any potential $V$, there is a critical temperature
$0\leq T_c (\lambda V)< \infty$ that separates two phases, a {\em
superfluid} phase for $ 0 \leq T < T_c(\lambda V)$ from a {\em
normal} phase for $ T_c(\lambda V) \leq T < \infty$. Note that $T_c(\lambda
V) = 0$ means that there is no superfluid phase for $\lambda V$.
Using the linear criterion (\ref{crittemp}) we can classify the potentials for which
$T_c(\lambda V)
> 0$, and simultaneously  we can evaluate the
asymptotic behavior of $T_c(\lambda V)$ in the limit of small
$\lambda$. This can be done by spectral theoretical methods. Applying the Birman-Schwinger principle one observes that
the critical temperature $T_c$ can be characterized by the
fact that the compact operator
\begin{equation}\label{eq:bsop}
\lambda (\sgn V) |V|^{1/2} K_{T_c,\mu}^{-1} |V|^{1/2}
\end{equation}
has $-1$ as its lowest eigenvalue. This operator is singular for $T_c
\to 0$, and the key observation is that its singular part is
represented by the operator $\lambda \ln(1/T_c)\V_\mu$, where $\V_\mu:
\, L^2(\Omega_\mu) \mapsto L^2(\Omega_\mu)$ is given by
\be\label{defvm} \big(\V_\mu u\big)(p) = \frac 1{(2\pi)^{3/2}} \frac
1{\sqrt{\mu}}\int_{\Omega_\mu}\hat V(p-q) u(q) \,d\omega(q)\,.  \ee
Here, $\Omega_\mu$ denotes the $2$-sphere with radius $\sqrt{\mu}$,
and $d\omega$ denotes Lebesgue measure on $\Omega_\mu$. We note that
the operator $\V_\mu$ has appeared already earlier in the literature
\cite{BY,LSW}.

Our analysis here is somewhat similar in spirit to the one concerning
the lowest eigenvalue of the Schr\"odinger operator $p^2+\lambda V$ in
\emph{two} space dimensions \cite{simon}. This latter case is
considerably simpler, however, as $p^2$ has a unique minimum at $p=0$,
whereas $K_{T,\mu}(p)$ takes its minimal value on the Fermi sphere
$p^2=\mu$, meaning that its minimum is highly degenerate. Hence, in
our case, the problem is reduced to analyzing a map from the $L^2$
functions on the Fermi sphere $\Omega_{\mu}$ (of radius $\sqrt{\mu}$)
to itself. Let us denote the lowest eigenvalue of $\V_\mu$ as
$$e_\mu(V) := \infspec  \V_\mu\,.$$
Whenever this eigenvalue is negative then the critical temperature
is non zero for {\em all} $\lambda > 0$, and we can evaluate its
asymptotics. Moreover, the converse is \lq\lq almost\rq\rq\ true:

\begin{theorem}[{\bf \cite[Theorem 1]{FHNS}}]\label{thm2.2}
  Let $V\in L^{3/2}(\R^3)\cap L^1(\R^3)$ be real-valued, and let
  $\lambda>0$.
\begin{itemize}
\item[(i)] Assume that $e_{\mu}(V)<0$. Then $T_c(\lambda V)$ is non-zero for
all $\lambda >0$, and
\begin{equation}\label{symptbeh}
\lim_{\lambda\to 0}   \lambda\, \ln \frac{\mu}{T_c(\lambda V)} =
-\frac{1}{e_{\mu}(V)} \,.
\end{equation}
\item[(ii)] Assume that $e_{\mu}(V)= 0$. If $T_c(\lambda V)$ is
  non-zero, then $\ln (\mu/{T_c(\lambda V)})\geq c \lambda^{-2}$ for some $c>0$ and
  small $\lambda$.
\item[(iii)] If there exists an $\epsilon>0$ such that $e_{\mu}(V-\epsilon|V|)= 0$,
then $T_c(\lambda V) = 0$ for small enough $\lambda$.
\end{itemize}
\end{theorem}
As we see, the occurrence of superfluidity as well as the asymptotic
behavior of $T_c(\lambda V)$ is governed by $e_\mu(V)$. A sufficient
condition for $e_\mu(V)$ to be negative is $ \int V < 0$. But one can
easily find other examples.  Eq.~\eqref{symptbeh} shows that the
critical temperature behaves like $ T_c(\lambda V) \sim \mu e^{
  1/(\lambda e_{\mu}(V))}$. In other words it is exponentially small
in the coupling.

In the following, we shall derive the second order correction, i.e.,
we will compute the constant in front of the exponentially small
term in $T_c$. For this purpose, we define an operator $\W_\mu$ on
$L^2(\Omega_\mu)$ via its quadratic form
\begin{align}\nonumber
  \langle u | \W_\mu |u\rangle = \int_{0}^\infty d|p| & \left( \frac
    {|p|^2}{\big||p|^2-\mu\big|} \left[ \int_{\S^2} d\Omega \left(
        |\hat\varphi(p)|^2 -
        |\hat\varphi(\sqrt\mu p/|p|)|^2 \right)\right] \right. \\
  \label{defW} & \quad \left. + \frac 1{|p|^2} \int_{\S^2} d\Omega\,
    |\hat\varphi(\sqrt\mu p/|p|)|^2\right) \,.
\end{align}
Here, $\hat\varphi(p) = (2\pi)^{-3/2} \int_{\Omega_\mu} \hat V(p-q)
u(q) d\omega(q)$, and $(|p|,\Omega)\in \R_+\times \S^2$ denote
spherical coordinates for $p\in\R^3$. We note that since $V\in
L^1(\R^3)$, $\int_{\S^2} d\Omega\,|\hat\varphi(p)|^2$ is Lipschitz
continuous in $|p|$ for any $u\in L^2(\R^3)$, and hence the radial
integration is well-defined, even in the vicinity of $p^2 = \mu$.  In
fact the operator $\W_\mu$ can be shown to be Hilbert-Schmidt class,
see \cite[Section 3]{HS}.

For $\lambda >0$, let
\begin{equation}\label{defB}
  \B_\mu = \lambda \frac \pi {2\sqrt \mu} \V_\mu
- \lambda^2 \frac{\pi}{2\mu} \W_\mu\,,
\end{equation}
and let $b_\mu(\lambda)$ denote its ground state energy,
\begin{equation}\label{defbm}
  b_\mu(\lambda) = \infspec \B_\mu \,.
\end{equation}
We note that if $e_\mu<0$, then also $b_\mu(\lambda)< 0$ for small
$\lambda$. In fact, if the eigenfunction corresponding to the lowest
eigenvalue $e_\mu$ of $\V_\mu$ is unique and equals $u\in
L^2(\Omega_\mu)$, then
\begin{equation}\label{bmexpl}
b_\mu(\lambda) = \langle u|\B_\mu|u\rangle + O(\lambda^3) = \lambda
\frac{ \pi e_\mu}{2\sqrt\mu} - \lambda^2 \frac{\pi \langle u|
  \W_\mu|u\rangle}{2\mu} + O(\lambda^3)\,.
\end{equation}
In the degenerate case, this formula holds if one chooses $u$ to be
the eigenfunction of $\V_\mu$ that yields the largest value $\langle
u|\W_\mu|u\rangle$ among all such (normalized) eigenfunctions.

With the aid of $b_\mu(\lambda)$, we can now recover the next order
of the critical temperature for small $\lambda$.

\begin{theorem}[{\bf \cite[Theorem 1]{HS}}] \label{constant}
  Let $V\in L^1(\R^3)\cap L^{3/2}(\R^3)$ and let $\mu>0$.  Assume that
  $e_\mu = \infspec \V_\mu <0$, and let $b_\mu(\lambda)$ be defined in
  (\ref{defbm}).  Then the critical temperature $T_c$ for the BCS
  equation is strictly positive and satisfies
  \begin{equation}\label{themeq}
    \lim_{\lambda \to 0} \left(\ln\left(\frac\mu {T_c}\right) +
      \frac {\pi}{2 \sqrt{\mu}\, b_\mu(\lambda)}\right) = 2 - \gamma - \ln(8/\pi)\,.
  \end{equation}
  Here, $\gamma\approx 0.577$ denotes Euler's constant.
\end{theorem}

The Theorem says that, for small $\lambda$,
\begin{equation}\label{formula}
  T_c \sim \mu \frac{8 e^{\gamma-2}}{\pi} e^{\pi/(2 \sqrt{\mu} b_\mu(\lambda))} \,.
\end{equation}
Note that $b_\mu(\lambda)$ can be interpreted as a (renormalized)
effective scattering length of $2\lambda V(x)$ (in second order Born
approximation) for particles with momenta on the Fermi sphere. In
fact, if $V$ is {\it radial} and $\int_{\R^3} V(x)dx < 0$, it is not
difficult to see that for small enough $\mu$ the (unique)
eigenfunction corresponding to the lowest eigenvalue $e_\mu$ of
$\V_\mu$ is the constant function $u(p)=(4\pi\mu)^{-1/2}$. (See
\cite[Section 2.1]{FHNS}.) For this $u$, we have
$$
\lim_{\mu\to 0} \langle u|\B_\mu|u\rangle = (\lambda/4\pi)\int_{\R^3} V(x) dx
- (\lambda/4\pi)^2 \int_{\R^6} \frac{V(x) V(y)}{|x-y|}dxdy \equiv
a_0(\lambda)\,.
$$
Here, $a_0(\lambda)$ equals the {\it scattering length} of $2\lambda
V$ in second order Born approximation. Assuming additionally that
$V(x)|x| \in L^1$ and bearing in mind that $b_\mu(\lambda) = \langle
u|\B_\mu|u\rangle + O(\lambda^3)$ for small enough $\mu$,
 we can, in fact, estimate the difference between $b_\mu(\lambda)$ and $a_0(\lambda)$. Namely we prove in \cite[Proposition 1]{HS} that
$$
  \lim_{\mu\to 0} \frac 1{\sqrt\mu}\left(\frac 1{\langle u|\B_\mu|u\rangle} - \frac 1
    {a_0(\lambda)}\right)=0\,.
$$
This yields the approximation
$$T_c \approx \mu \frac{8 e^{\gamma-2}}{\pi} e^{\pi/(2 \sqrt{\mu} a_0(\lambda))}
$$
in the limit of  small $\lambda$ {\it and} small $\mu$. This
expression is well-known in the physics literature \cite{gorkov,NS}.
We point out, however, that our formula (\ref{formula}) is much more
general since it holds for any value of $\mu>0$.

\subsection{Energy Gap at Zero Temperature}

Consider now the zero temperature case $T=0$. In this case, it is
natural to formulate a functional depending only on $\alpha$ instead
of $\gamma$ and $\alpha$. In fact, for $T=0$  the optimal choice of
$\gamma(p)$ in $\F_T$ for given $\hat\alpha(p)$ is clearly
\begin{equation}\label{gal}
  \gamma(p) = \left\{ \begin{array}{ll}
      \half (1+\sqrt{1-4|\hat\alpha(p)|^2}) & {\rm for\ } p^2< \mu \\
      \half (1-\sqrt{1-4|\hat\alpha(p)|^2}) & {\rm for\ } p^2>\mu
    \end{array}\right.\,.
\end{equation}
Subtracting an unimportant constant, this leads to the {\it zero temperature
BCS functional}
\begin{equation}\label{deffa}
  \F_0(\alpha)
  =\frac 12 \int_{\R^3} |p^2-\mu|\left(1-\sqrt{1-4|\hat\al(p)|^2}\right)dp+ \lambda \int_{\R^3}
  V(x)|\alpha(x)|^2\,dx\,.
\end{equation}

The variational equation satisfied by a minimizer of (\ref{deffa}) is then
\begin{equation}\label{bcset}
  \Delta(p) = -\frac \lambda{(2\pi)^{3/2}} \int_{\R^3} \hat V(p-q)
  \frac{\Delta(q)}{E(q)} \, dq\,,
\end{equation}
with $\Delta(p) =  2E(p) \hat \alpha(p)$. 
This is simply the BCS equation (\ref{bcse}) at $T=0$.  For a solution
$\Delta$, the {\it energy gap} $\Xi$ is defined as
\begin{equation}\label{defxi}
\Xi = \inf_p E(p) = \inf_p \sqrt{(p^2-\mu)^2 + |\Delta(p)|^2}\,.
\end{equation}
It has the interpretation of an energy gap in the corresponding
second-quantized BCS Hamiltonian (see, e.g., \cite{MR} or the
appendix in \cite{HHSS}.)

A priori, the fact that the order parameter $\Delta$ is non vanishing
does not imply that $\Xi>0$. Strict positivity of $\Xi$ turns out to
be related to the continuity of the corresponding $\gamma$ in
(\ref{gal}). In fact, we are going to prove in Lemma \ref{lemequ} that
if $V$ decays fast enough, i.e., $V(x)|x| \in L^{6/5}(\R^3)$, the two
properties, $\Xi > 0$ and $\gamma(p)$ continuous, are equivalent. Both properties hold true under the assumption that $\int V < 0$:

\begin{theorem}\label{psofgap}
  Let $V \in L^{3/2}\cap L^1$, with $ V(x)|x| \in L^{6/5}(\R^3) $ and
  $\int V = (2\pi)^{3/2}\hat V(0) < 0$. Let $\alpha$ be a minimizer of the BCS
  functional. Then $\Xi$ defined in (\ref{defxi}) is strictly
  positive, and the corresponding momentum distribution $\gamma$ in
  (\ref{gal}) is continuous.
\end{theorem}

One of the difficulties involved in evaluating $\Xi$ is the potential
non-uniqueness of minimizers of (\ref{deffa}), and hence
non-uniqueness of solutions of the BCS gap equation (\ref{bcset}). The
gap $\Xi$ may depend on the choice of $\Delta$ in this case. For
potentials $V$ with non-positive Fourier transform, however, we can
prove the uniqueness of $\Delta$ and, in addition, we are able to
derive the precise asymptotic of $\Xi$ as $\lambda\to 0$.

In the following we will restrict our attention to radial potentials
$V$ with non-positive Fourier transform.  We also assume that $\hat
V(0) = (2\pi)^{-3/2} \int V(x) dx < 0$.  It is easy to see that
$e_\mu = \infspec \V_\mu <0$ in this case, and that the
(unique) eigenfunction corresponding to this lowest eigenvalue of
$\V_\mu$ is the constant function.

In particular we have the following asymptotic behavior of the
energy gap $\Xi$ as $\lambda \to 0$.

\begin{theorem}[{\bf \cite[Theorem 2]{HS}}]\label{gap}
  Assume that $V\in L^1(\R^3)\cap L^{3/2}(\R^3)$ is radial, with $\hat
  V(p)\leq 0$ and $\hat V(0)<0$. Then there is a unique minimizer (up
  to a constant phase) of the BCS functional (\ref{deffa}) at
  $T=0$. The corresponding energy gap,
$ \Xi = \inf_p \sqrt{ (p^2-\mu)^2 + |\Delta(p)|^2}\,, $ is strictly
positive, and satisfies
\begin{equation}\label{themeq2}
  \lim_{\lambda \to 0} \left(\ln\left(\frac\mu \Xi \right) +
    \frac {\pi}{2 \sqrt{\mu}\, b_\mu(\lambda)}\right) = 2 - \ln(8)\,.
\end{equation}
Here, $b_\mu(\lambda)$ be defined in (\ref{defbm}).
\end{theorem}
The Theorem says that, for small $\lambda$,
$$
\Xi \sim \mu \frac{8}{e^2} e^{\pi/(2 \sqrt{\mu} b_\mu(\lambda))} \,.
$$
In particular, in combination with Theorem~\ref{constant}, we obtain
the {\em universal ratio}
$$
\lim_{\lambda\to 0} \frac{ \Xi}{T_c} = \frac \pi{e^\gamma} \approx
1.7639\,.
$$
That is, the ratio of the energy gap $\Xi$ and the critical
temperature $T_c$ tends to a universal constant as $\lambda\to 0$,
independently of $V$ and $\mu$. This property has been observed
before for the original BCS model with rank one interaction
\cite{BCS,MR}, and in the low density limit for more general
interactions \cite{gorkov} under additional assumptions. Our
analysis shows that it is valid in full generality at small coupling
$\lambda \ll 1$.

\section{Sketch of the proof of Theorem \ref{mthmbcs}}

The backbone of our analysis is the linear criterion in Theorem
\ref{mthmbcs}. As a first step towards its proof, one has to prove that the functional
$\F_T(\gamma, \alpha)$ in \eqref{freeenergy} attains a minimum  on the set
$$ \de = \{ (\gamma,\alpha) \,| \, \gamma \in L^1(\R^3,(1+p^2) dp), \alpha \in
H^1(\R^3), 0\leq \gamma \leq 1, |\hat \alpha|^2 \leq \gamma(1-\gamma)
\}.$$ This can be done by proving lower semi-continuity of $\F_T$ on
$\de$. See \cite[Prop.~1]{HHSS} for details.  Theorem \ref{mthmbcs} is
then a direct consequence of the equivalence of the following three
statements \cite[Theorem 1]{HHSS}:
\begin{itemize}
\item[(i)]
The normal state $(\gamma_0,0)$, with $\gamma_0 = [e^{(p^2 - \mu)/T}
+ 1]^{-1}$ being the Fermi-Dirac distribution, is unstable under
pair formation, i.e.,
\begin{equation*}
\inf_{(\gamma,\alpha)\in\de} \F_T (\gm,\al) <
\F_T(\gm_0,0)\,.
\end{equation*}
\item[(ii)] There exists a pair $(\gamma,\alpha)\in \de$, with $\alpha
  \neq 0$, such that
\begin{equation}\label{defed}
\Delta(p)=\dfrac{p^2-\mu}{\half - \gamma(p)}\hat\alpha(p)
\end{equation}
satisfies the BCS gap equation \eqref{bcse}.
\item[(iii)] The linear operator
$ K_{T,\mu} +  V$
 has at least one negative eigenvalue.
\end{itemize}

The proof of the equivalence of these three statement consists of the
following steps. First, it is straightforward to show that (i)
$\Rightarrow$ (ii). By evaluating the stationary equations in both
variables, $\gamma$ and $\alpha$, one shows that the combination
\eqref{defed} satisfies the BCS equation \eqref{bcse}.

To show that (iii) $\Rightarrow$ (i), 
first note that $(\gamma_0,0)$  is the minimizer of
$\F_T$ in the case $V=0$. Consequently $\frac d{dt} \F_T(\gamma_0, t
g)|_{t=0} = 0$ for general $g$. Moreover, a simple calculation shows that 
$$ \frac {d^2}{dt^2} \F(\gamma_0, tg)_{t=0} = 2 \bra g|K_{T,\mu} +
\lambda V|g \ket\,.$$ If $K_{T,\mu} + \lambda V$ has a negative
eigenvalue, we thus see that $\F(\gamma_0,t g)<
\F_T(\gamma_0,0)$ for small $t$  and an appropriate choice of $g$.

The hardest part in showing the equivalence of the three statements is
to show that (ii) $\Rightarrow$ (iii). Given a pair $(\tilde
\gamma,\tilde \alpha)$ such that the corresponding $\Delta$ in
\eqref{defed} satisfies the BCS equation (\ref{bcse}), we note that if
$\hat \alpha = m(p) \hat {\tilde \alpha}(p)$ and $\gamma(p) = 1/2 +
m(p)(\gamma(p) -1/2)$, the pair $(\gamma,\alpha)$ yields the same $\Delta$ and hence also satisfies (\ref{bcse}). Moreover, with the choice 
$$m(p) = \frac{p^2 - \mu}{\frac 12 - \tilde
\gamma(p)}\frac{\tanh\frac{E(p)}{2T}}{2E(p)}$$ 
(where $E(p) = \sqrt{(p^2 - \mu)^2 + |\Delta(p)|^2}$), 
the new pair $(\gamma,
\alpha)$ satisfies additionally
\begin{align}\label{EQ1}
\frac{2E(p)}{\tanh\frac{E(p)}{2T}} = &  \frac{p^2 - \mu}{\frac 12 - \gamma(p)} \\
\label{E2}
\frac{\lambda}{(2\pi)^{3}} \int \hat V(p-q) \hat \alpha(q) dq  = &-\frac{p^2 - \mu}{ \frac 12 - \gamma} \hat  \alpha(p)\,.
\end{align}
Note that in the case
$V=0$, i.e.,  $\Delta = 0$, the equation \eqref{EQ1} reduces to
$$ 2 K_{T,\mu}(p) =  \frac{p^2 - \mu}{\frac 12 - \gamma_0}.$$
Using this fact, together with \eqref{E2}, we thus obtain 
\be\label{E3}
\langle \alpha |  K_{T,\mu} + \lambda V |\alpha\rangle = \frac 12
\left\langle \alpha \left| \frac{p^2 - \mu}{\frac 12-\gamma_0} -
\frac{p^2 - \mu}{\frac 12 - \gamma} \right|\alpha \right\rangle\,.\ee 
Using the definition of $E(p)$ and the strict monotonicity of the function
$x\mapsto  x/{\tanh\frac{x}{2T}} $ for  $x \geq 0$,  we infer from
\eqref{EQ1} that
$$ \frac{p^2 - \mu}{\frac 12 -\gamma_0} \leq \frac{p^2 - \mu}{\frac 12 - \gamma}\,,$$ with strict 
the inequality on the set where $\Delta \neq 0$.  Consequently, the
expression $\eqref{E3} $ is strictly negative. Hence $K_{T,\mu} +
\lambda V$ has a negative eigenvalue. This shows that (ii) implies (iii).

\section{Proof of Theorems \ref{thm2.2} and \ref{constant}}\label{proof:thm1}

For a (not necessarily sign-definite) potential $V(x)$ let us use the
notation
\begin{equation*}
  V(x)^{1/2} = (\sgn V(x)) |V(x)|^{1/2} \,.
\end{equation*}
From our definition of the critical temperature $T_c$ it follows
immediately that for $T=T_c$ the operator $ K_{T,\mu} + \lambda V$
has and eigenvalue $0$ and no negative eigenvalue. If $\psi$ is the
corresponding eigenvector, one can rewrite the eigenvalue equation
in the form
$$ -\psi = \lambda K_{T,\mu}^{-1} V \psi .$$
Multiplying this equation by $V^{1/2}(x)$, one obtains an eigenvalue
equation for $\varphi = V^{1/2} \psi$. This argument works in both
directions and is called the Birman-Schwinger principle (see
\cite[Lemma 1]{FHNS}). In particular it  tells us that the critical
temperature $T_c$ is determined by the fact that for this value of
$T$ the smallest eigenvalue of
\begin{equation}\label{defofbt}
  B_T = \lambda V^{1/2}K_{T,\mu}^{-1}|V|^{1/2}
\end{equation}
equals $-1$. Note that although $B_T$ is not self-adjoint, it has real
spectrum.

Let $\dig: L^1(\R^3) \to L^2(\Omega_\mu)$ denote the (bounded)
operator which maps $\psi\in L^1(\R^3)$ to the Fourier transform of
$\psi$, restricted to the sphere $\Omega_\mu$. Since $V\in L^1(\R^3)$,
multiplication by $|V|^{1/2}$ is a bounded operator from $L^2(\R^3)$ to
$L^1(\R^3)$, and hence $\dig |V|^{1/2}$ is a bounded operator from
$L^2(\R^3)$ to $L^2(\Omega_\mu)$. Let
$$
m_\mu(T) = \max\left\{ \frac 1{4\pi \mu} \int_{\R^3} \left( \frac 1{K_{T,\mu}(p)}
  - \frac 1{p^2}\right) dp\, , 0\right\} \,,
$$
and let
\begin{equation}\label{defmt}
  M_T = K_{T,\mu}^{-1} - m_\mu(T) \dig^* \dig\,.
\end{equation}
As in \cite[Lemma~2]{FHNS} one can show that $ V^{1/2}
M_T |V|^{1/2}$ is a Hilbert-Schmidt operator on $L^2(\R^3)$, and its
Hilbert Schmidt norm is bounded uniformly in $T$. In particular,  the
singular part of $B_T$ as $T\to 0$ is entirely determined by
$V^{1/2}\dig^* \dig |V|^{1/2}$.

Since $V^{1/2} M_T |V|^{1/2}$ is uniformly bounded, we can choose
$\lambda$ small enough such that  $1+\lambda V^{1/2} M_T |V|^{1/2}$
is invertible, and we can then write $1+ B_T$ as
\begin{align}\label{1ba}
  1+ B_T &= 1+ \lambda V^{1/2} \left( m_\mu(T) \dig^* \dig +
    M_T\right) |V|^{1/2} \\ \nonumber &= \left(1+ \lambda V^{1/2} M_T
    |V|^{1/2} \right) \left( 1 + \frac{\lambda m_\mu(T)}{1+ \lambda
      V^{1/2} M_T |V|^{1/2}} V^{1/2} \dig^* \dig |V|^{1/2}\right)\,.
\end{align}
Then $B_T$ having an eigenvalue $-1$  is equivalent to
\begin{equation}\label{a1}
  \frac{\lambda m_\mu(T)}{1+ \lambda V^{1/2} M_T |V|^{1/2}} V^{1/2}
  \dig^* \dig |V|^{1/2}
\end{equation}
having an eigenvalue $-1$. The operator in (\ref{a1}) is
isospectral to the selfadjoint operator
\begin{equation}\label{b1}
\dig |V|^{1/2} \frac {  \lambda m_\mu(T) }{1+ \lambda V^{1/2} M_T
    |V|^{1/2}} V^{1/2} \dig^*\,,
\end{equation}
acting on $L^2(\Omega_\mu)$.

At $T=T_c$, $-1$ is the smallest eigenvalue of $B_T$, hence
(\ref{a1}) and (\ref{b1}) have an eigenvalue $-1$ for this value of
$T$. Moreover, we can conclude that $-1$ is actually the {\it
smallest} eigenvalue of (\ref{a1}) and (\ref{b1}) in this case. For,
if there were an eigenvalue less then $-1$, we could increase $T$
and, by continuity, find some $T>T_c$ for which there is an
eigenvalue $-1$. Using (\ref{1ba}), this would contradict the fact
that $B_T$ has no eigenvalue $-1$ for $T>T_c$.

Consequently, the equation for the critical temperature can be written
as \be\label{equcrittemp} \lambda m_\mu(T_c) \, \infspec \dig
|V|^{1/2} \frac { 1}{1+ \lambda V^{1/2} M_{T_c} |V|^{1/2}} V^{1/2}
\dig^* = -1\,. \ee This equation is the starting point for the proof of
Theorems~\ref{thm2.2} and~\ref{constant}.

\begin{proof}[Proof of Theorem \ref{thm2.2}]
  Up to first order in $\lambda$ the equation \eqref{equcrittemp}
  reads \be\label{crittempfirst} \lambda m_\mu(T_c)\, \infspec \dig
  [V- \lambda V M_{T_c}V + O(\lambda^2) ]\dig^* = -1\,, \ee where the
  error term $O(\lambda^2)$ is uniformly bounded in $T_c$.  Note that
  $\dig V \dig^* = \sqrt\mu\, \V_\mu$ defined in (\ref{defvm}). Assume
  now that $e_\mu = \infspec \V_\mu$ is strictly negative. Since
  $V^{1/2} M_{T_c} V^{1/2}$ is uniformly bounded, it follows
  immediately that
$$
\lim_{\lambda \to 0} \lambda m_\mu(T_c) =- \frac 1{\infspec \dig V
  \dig^*} = - \frac 1 {\sqrt\mu \, e_\mu}\,.
$$
Together with the asymptotic behavior $m_\mu(T) \sim \mu^{-1/2}
\ln(\mu/T)$ as $T\to 0$, this implies the leading order behavior of
$\ln (\mu/T_c)$ as $\lambda\to 0$ and proves the statement in $(i)$.

In order to see $(ii)$ it suffices to realize that, in the case
$\dig V \dig^* \geq 0$, Eq.~\eqref{crittempfirst} yields $
m_{T_c}\geq \const/\lambda^2.$

The statement $(iii)$ is a consequence of the fact that
$$ \dig |V|^{1/2} \frac {
1}{1+ \lambda V^{1/2} M_{T_c}
    |V|^{1/2}} V^{1/2} \dig^* \geq \dig [V - \const \lambda |V|] \dig^* \geq 0,
    $$ for $\lambda$ small enough. We refer to \cite{FHNS} for details.
\end{proof}

\begin{proof}[Proof of Theorem \ref{constant}]
To obtain the next order, we use Eq.~\eqref{crittempfirst} and
employ first order perturbation theory. Since $\dig V \dig^*$ is
compact and $\infspec \dig V \dig^*<0$ by assumption, first order
perturbation theory implies that
\begin{equation}\label{deno}
  m_\mu(T_c) = \frac {-1}{ \lambda \langle u| \dig V \dig^*| u\rangle -
    \lambda ^2 \langle u| \dig V M_{T_c} V \dig^*| u\rangle + O(\lambda^3)}\,,
\end{equation}
where $u$ is the (normalized) eigenfunction corresponding to the lowest eigenvalue
of $\dig V \dig^*$. (In case of degeneracy, one has to the choose the
$u$ that minimizes the $\lambda^2$ term in the denominator of
(\ref{deno}) among all such eigenfunctions.)

Eq.~(\ref{deno}) is an implicit equation for $T_c$. Since $\dig V
M_T V\dig^*$ is uniformly bounded and $T_c\to 0$ as $\lambda \to 0$,
we have to evaluate the limit of $\langle u| \dig V M_T V \dig^*|
u\rangle$ as $T\to 0$. To this aim, let $\varphi = V \dig^* u$. Then
\begin{align}\nonumber
  &\langle u| \dig V M_T V \dig^* |u\rangle \\ \label{comeq}
& = \int_{\R^3} \frac 1{K_{T,\mu}(p)}
  |\hat\varphi(p)|^2 \, dp - m_\mu(T) \int_{\Omega_\mu}
  |\hat\varphi(p)|^2 \, d\omega(p) \\ \nonumber & = \int_{\R^3}
  \left(\frac 1{K_{T,\mu}(p)} \left[ |\hat\varphi(p)|^2
      -|\hat\varphi(\sqrt\mu p/|p|)|^2 \right] + \frac 1{p^2}
    |\hat\varphi(\sqrt\mu p/|p|)|^2 \right) dp\,.
\end{align}
Recall that $K_{T,\mu}(p)$ converges to $|p^2-\mu|$ as $T\to 0$.
Using the Lipschitz continuity of the spherical average of
$|\hat\varphi(p)|^2$ (see \cite[Eq.~(29)]{HS}) it is easy to
see that
\begin{equation}\label{deno2}
  \lim_{T\to 0} \langle u| \dig V M_T V \dig^*| u\rangle  = \langle u| \W_\mu |u \rangle\,,
\end{equation}
with $\W_\mu$ defined in (\ref{defW}). In particular, combining
(\ref{deno}) and (\ref{deno2}), we have thus shown that
\begin{equation}\label{denof}
  \lim_{\lambda\to 0} \left( m_\mu(T_c) + \frac 1{\infspec \left
        (\lambda \sqrt \mu\, \V_\mu - \lambda^2 \W_\mu\right)} \right) =
  0\,.
\end{equation}
The statement follows by using the   asymptotic behavior (\cite[Lemma 1]{HS})
  \begin{equation}\label{lemresult}
    m_\mu(T) =  \frac 1{\sqrt\mu}
\left( \ln \frac \mu T + \gamma-2 + \ln\frac 8\pi  + o(1)\right)
  \end{equation}
  in the limit of small $T$, where $\gamma\approx 0.5772$ is Euler's
  constant.
\end{proof}

\section{Proof of Theorems~\ref{psofgap} and~\ref{gap}}

\subsection{Sufficient condition for $\Xi > 0$}

If $e_\mu(V) < 0$ we know that the BCS equation \eqref{bcset} has a
solution, meaning the system shows a superfluid phase for $T=0$.  This
is not sufficient, however, to guarantee the existence of a positive
gap $\Xi > 0$ nor the continuity of the momentum distribution
$\gamma$.  Unlike the case of the critical temperature, we lack a
linear criterion which allows a precise characterization of potentials
$V$ giving rise to a strictly positive gap. We are, however, able
to derive sufficient conditions, namely a fast enough decay of
$V$. Under such assumptions one can show the equivalence of the positivity of $\Xi$ and the continuity of $\gamma$. Both hold true if
additionally $\int V < 0$. It remains an open problem  to find
examples for $V$ such that $e_\mu < 0$ but $\Xi = 0$.

\begin{lemma}\label{lemequ}
Assume that  $V \in L^{3/2}$ and that $ V(x) |x| \in L^{6/5}(\R^3)$. Then $ \Xi
> 0 $ if and only if $\gamma $ is continuous.
\end{lemma}

\begin{proof}
  It is easy to deduce \cite{HHSS} from the BCS equation (\ref{bcset}) that $\hat
  \alpha$ is in $C^0(\R^3)$. Because of \eqref{gal} the continuity of
  $\gamma$ is equivalent to the fact that $|\hat \alpha | \equiv 1/4$
  on the Fermi $\Omega_\mu$. From the relation $\Delta(p) = 2 E(p) \hat \alpha(p)$
  one obtains
\begin{equation}\label{nl}
|\hat \alpha (p)|^2 = \frac 14 \frac 1{\sqrt{\frac{(p^2 - \mu)^2}{|\Delta(p)|^2} + 1 }}\,,
\end{equation}
and we can conclude that $|\hat \alpha|^2 = 1/4$ on the Fermi
surface if and only if $\Delta(p)$ does not vanish on $\Omega_\mu$.
Namely, suppose that $\Delta$ vanishes at some $p'$ on the Fermi
surface. Since $\alpha \in H^1(\R^3)$ we see that  $\alpha \in
L^2(\R^3)\cap L^{6}(\R^3)$ and hence, together with $V(x)|x| \in
L^{6/5}$, H\"older's inequality implies that $\check{\Delta}(x)|x| = 
V(x)\alpha(x)|x| \in
L^1(\R^3)$. We thus infer that $\Delta(p)$ is Lipschitz continuous,
meaning that $\Delta(p)$ cannot decay slower to $0$ than linear.
Hence there is a $\delta$ such that $\lim_{p \to p'} \frac{(p^2 -
\mu)^2}{|\Delta(p)|^2} \geq \delta$ and  $ |\alpha (p')|^2 \leq
\frac 14 \frac 1{\sqrt{\delta + 1}} < \frac 14.$
\end{proof}

\begin{proof}[Proof of Theorem \ref{psofgap}]
Let $\alpha$ be a global minimizer of the BCS functional
$\F_0$. Then for any $\hat g \in C^\infty_0(\R^3)$ such that
$|\hat \alpha + \epsilon \hat g|\leq 1/2$ for $\epsilon$ small enough,
\begin{equation}\label{condminim}
\left. \frac {d^2}{d\epsilon^2} \F(\alpha + \eps g)\right|_{\eps = 0} \geq 0.
\end{equation}
A straightforward calculation yields
\begin{equation}\label{2ndfunctderiv}
\left. \frac {d^2}{d\epsilon^2} \F(\alpha + \eps g)\right|_{\eps = 0} = 2 \langle
g| E(-i\nabla) + \lambda V |g\rangle + 8\int \frac{ |p^2 -
\mu|[\Re(\hat \alpha \bar{\hat g})]^2}{[1 - 4|\hat
\alpha|^2]^{3/2}}.
\end{equation}
Assume now that $\Xi = 0$. This means that $\Delta$ has to vanish at
some point $p' \in \Omega_\mu$. Then there has to be an open
neighborhood on $\Omega_\mu$ on which $\Delta$ vanishes. In fact,
according to the argument in the proof of Lemma \ref{lemequ}
(Eq.~\eqref{nl} and Lipschitz continuity of $\Delta$) there is a
neighborhood $\mathcal {N}_\delta (p')\subset \R^3$ in the vicinity of
$p'$ where $|\hat \alpha |^2 < 1/4 - \delta$ for some $\delta >0$, and
hence $\Delta$ vanishes on $\mathcal {N}_\delta(p') \cap
\Omega_\mu$. Note that $\Delta$ cannot vanish at one point on the
Fermi surface since otherwise $|\hat \alpha| = 1/2$ except on one
point, which contradicts the continuity of $\hat \alpha$.

We shall now construct an appropriate trial sequence $\hat g_n$,
essentially supported in $\mathcal{N}_\delta$, such that \be\label{trialkin}
\lim_{n\to\infty} \left[ \langle g_n| E(-i\nabla)|g_n \rangle + 8\int
  \frac{ |p^2 - \mu|[\Re(\hat \alpha \bar{\hat {g_n}})]^2}{[1 - 4|\hat
    \alpha|^2]^{3/2}} \right] = 0 \ee and \be\label{trialpot}
\lim_{n\to \infty} \langle g_n| V |g_n\rangle = \int_{\R^3} V(x) dx <
0\,.  \ee 
This gives a contradiction to
\eqref{condminim}.

For the construction of $g_n$ let $\psi_n \in L^2(\Omega_\mu)$ be
supported in $\mathcal{N}_\delta(p') \cap \Omega_\mu$ such that
$\psi_n(s) \to \delta(s-p')$ as $n\to \infty$. Choose also $f_n
\in L^2(\R_+, t^2 dt)$ such that $f_n(t) \to \delta(\sqrt\mu-t)$,
and let $\hat g_n(p) = \psi_n(s)f_n(|p|)$. Observe that on
$\mathcal{N}(p')$, $E(p) =|p^2-\mu|/\sqrt{1-4|\hat\alpha(p)|^2} \leq c
|p^2-\mu|$ for some constant $c$, and thus grows {\it linearly} in
$|p|$ close to $\sqrt\mu$. Hence one easily sees that the problem here
is equivalent to the existence of a negative eigenvalue of the
relativistic operator $|p| + V$ in {\it one} dimension. Using the
Birman-Schwinger principle, it is easy to see that the latter always
has a negative eigenvalue if $\int V < 0$.
\end{proof}

\subsection{Proof of Theorem \ref{gap}}

The energy gap of the system at zero temperature, $\Xi = \inf_p E(p)$, with
$$E(p)=|p^2-\mu|/\sqrt{1-4|\hat\alpha(p)|^2}=\sqrt{|p^2-\mu|^2+|\Delta(p)|^2},$$
depends on the behavior of $|\Delta(p)|$ on the Fermi sphere. The
function $\Delta$ is not unique, in general and need not be radial
even in case $V$ is radial.

Under the assumption that $\hat V$ is non-positive and $\hat V(0) <
0$, we shall argue in the following that the
minimizer of the BCS functional (\ref{deffa}) at $T=0$ is unique
\cite[Lemma 3]{HS}.  If, in addition, $V$ is radial, this necessarily
implies that also the minimizer has to be radial. Since $\hat V \leq
0$,
  \begin{equation}\label{sq}
    \int_{\R^6} \overline{\hat\alpha(p)} \hat V(p-q) \hat\alpha(q) \, dpdq \geq
\int_{\R^6} |\hat\alpha(p)| \hat V(p-q)| \hat\alpha(q)|\, dpdq \,.
  \end{equation}
  Hence, if $\hat\alpha(p)$ is a minimizer of $\F_0$, \eqref{deffa}, so is $|\hat \alpha(p)|$.

  Assume now there are two different minimizers $f\neq
  g$, both with nonnegative Fourier transform.  Since $t \to 1 -
  \sqrt{1 - 4 t}$ is strictly convex for $0\leq t \leq 1/2$ we
  see that $\psi = \frac 1{\sqrt{2}} f + i \frac 1{\sqrt{2}} g$,
  satisfies
$$ \F_0 (\psi)  < \half \F_0(f) + \half \F_0(g)\,.$$
This is a contradiction to $f,g$ being distinct minimizers, and hence
$f=g$. In particular, the absolute value of a minimizer has to be
unique. If $\hat\alpha$ is the unique non-negative minimizer, then one
easily sees from the BCS equation (using $\int V<0$) that $\hat
\alpha$ is, in fact, strictly positive. Hence {\it any} minimizer is
non-vanishing. But (\ref{sq}) is {\it strict} for non-vanishing
functions, unless $\hat\alpha(p)=|\hat\alpha(p)| e^{i\kappa}$ for
some constant $\kappa\in\R$.

To summarize, we have just argued that for $\hat V \leq 0$, $\hat V(0)
< 0$ and $V$ radial, the solution of the BCS equation is unique, up to
a constant phase, and it is radially symmetric. This will enable us to
apply the same methods as we used for the critical temperature $T_c$
in order to derive the asymptotic behavior of $\Xi$.

The variational equation \eqref{bcset} for the minimizer
of $\F_0$ can be rewritten in terms of $\alpha$ as
\begin{equation}\label{evea}
\left(E(-i\nabla) + \lambda V(x) \right)\alpha(x) = 0.\,
\end{equation}
 That is, $\alpha$ is an eigenfunction of the
pseudodifferential operator $E(-i\nabla)+\lambda V(x)$, with zero
eigenvalue.  Since $\hat V \leq 0$ and $\hat \alpha(p)$ is
non-negative we can even conclude that $\alpha$ has to be the ground
state.

Similarly to the proof of Theorem~\ref{constant}, we can now employ
the Birman-Schwinger principle to conclude from \eqref{evea} that
$\phi_\lambda = V^{1/2} \alpha$ satisfies the eigenvalue equation
\begin{equation}\label{bse}
 \lambda V^{1/2}  \frac1
  {\sqrt{(p^2 - \mu)^2 + |\Delta(p)|^2}}|V|^{1/2} \phi_\lambda = - \phi_\lambda\,.
\end{equation}
Moreover, there are no eigenvalues smaller than
$-1$ of the operator on the left side of (\ref{bse}).

Let
\begin{equation}\label{defmtd}
  \widetilde m_\mu(\Delta) = \max\left\{ \frac 1{4\pi\mu} \int_{\R^3}\left(
\frac 1 {\sqrt{(p^2-\mu)^2 +
        |\Delta(p)|^2}} -\frac 1{p^2}\right) dp \, , \, 0\right\} \,.
\end{equation}
Similarly to (\ref{defmt}), we split the operator in (\ref{bse}) as
$$
V^{1/2} \frac 1{E(-i\nabla)} |V|^{1/2} = \widetilde m_\mu(\Delta) V^{1/2} \dig^*\dig
|V|^{1/2} + V^{1/2}M_\Delta |V|^{1/2}\,.
$$
Again one shows that $V^{1/2}M_\Delta |V|^{1/2}$ is bounded in
Hilbert-Schmidt norm, independently of $\Delta$. Moreover, as in the
proof of Theorem~\ref{constant} (cf.~Eqs.~(\ref{1ba})--(\ref{b1})),
the fact that the lowest eigenvalue of $\lambda V^{1/2}
E(-i\nabla)^{-1} |V|^{1/2}$ is $-1$ is, for small enough $\lambda$,
equivalent to the fact that the selfadjoint operator on
$L^2(\Omega_\mu)$
\begin{equation}\label{mbd}
 \dig |V|^{1/2} \frac {  \lambda
  \widetilde m_\mu(\Delta)}{1 + \lambda V^{1/2}
    M_\Delta |V|^{1/2} }V^{1/2} \dig^*
\end{equation}
has $-1$ as its smallest eigenvalue. This implies that
$\lim_{\lambda\to 0} \lambda \widetilde m_\mu(\Delta) =
-1/(\sqrt\mu\,e_\mu)$ and hence, in particular, $\widetilde
m_\mu(\Delta) \sim \lambda^{-1}$ as $\lambda \to 0$. The unique
eigenfunction corresponding to the lowest eigenvalue $e_\mu<0$ of
$\V_\mu$ is, in fact, a positive function, and because of radial
symmetry of $V$ it is actually the constant function
$u(p)=(4\pi\mu)^{-1/2}$.

We now give a precise characterization of $\Delta(p)$ for small
$\lambda$.

\begin{lemma}\label{Deltach}
  Let $V \in L^1\cap L^{3/2}$ be radial, with $\hat V \leq 0$ and
  $\hat V(0) < 0$, and let $\Delta$ be given in (\ref{bcset}), with
  $\alpha$ the unique minimizer of the BCS functional
  (\ref{deffa}). Then
  \begin{equation}
    \Delta(p) = - f(\lambda) \left( \int_{\Omega_\mu} \hat V(p-q)  \,
      d\omega(q)  + \lambda \eta_\lambda(p)\right)
  \end{equation}
  for some positive function $f(\lambda)$, with
  $\| \eta_\lambda\|_{L^\infty(\R^3)}$ bounded independently of $\lambda$.
\end{lemma}

\begin{proof}
  Because of (\ref{bse}), $\dig |V|^{1/2}\phi_\lambda$ is the
  eigenfunction of (\ref{mbd}) corresponding to the lowest eigenvalue
  $-1$. Note that because of radial symmetry, the constant function
  $u(p)=(4\pi\mu)^{-1/2}$ is an eigenfunction of (\ref{mbd}). For
  small enough $\lambda$ it has to be eigenfunction corresponding to
  the lowest eigenvalue (since it is the unique ground state of the
  compact operator $\dig V \dig^*$).  We conclude
  that
  \begin{equation}\label{combw}
    \phi_\lambda =  f(\lambda) \frac 1{1+\lambda V^{1/2} M_\Delta
      |V|^{1/2}} V^{1/2} \dig^* u  = f(\lambda) \left(
      V^{1/2}\dig^* u + \lambda \xi_\lambda\right)
  \end{equation}
  for some normalization constant $f(\lambda)$. Note that
  $\|\xi_\lambda\|_2$ uniformly bounded for small $\lambda$, since
  both $V^{1/2}M_\Delta |V|^{1/2}$ and $V^{1/2}\dig^*$ are bounded
  operators.

From  (\ref{evea}) and the definition $\phi_\lambda = V^{1/2}\alpha$
we know that
$$
\Delta(p) = 2 E(p) \hat\alpha(p) = - 2 \lambda \widehat {V
  \alpha}(p) = -2 \lambda \widehat{|V|^{1/2} \phi_\lambda}(p)\,.
$$
In combination with (\ref{combw}) this implies  that
$$
\Delta(p) = - 2\lambda f(\lambda) \left( \widehat{V \dig^* u}(p) + \lambda
  \widehat{\eta_\lambda}(p)\right)\,,
$$
with $\eta_\lambda = |V|^{1/2} \xi_\lambda$. With
$\|\widehat{\eta_\lambda}\|_\infty\leq (2\pi)^{-3/2}
\|\eta_\lambda\|_1 \leq (2\pi)^{-3/2} \|V\|_1 \|\xi_\lambda\|_2$ by
Schwarz's inequality, we arrive at the statement of the Lemma.
\end{proof}

With the aid of Lemma~\ref{Deltach} and Lipschitz continuity of $
\int_{\Omega_\mu} \hat V(p-q) \, d\omega(q)$ (which follows from $V\in
L^1(\R^3)$) it is not difficult to see that
\begin{equation}\label{calcmtp}
  \widetilde m_\mu(\Delta) = \frac 1{\sqrt\mu}\left( \ln \frac \mu {\Delta(\sqrt\mu)} -2 +
    \ln 8 + o(1)\right)
\end{equation}
as $\lambda\to 0$. From Eq.~(\ref{mbd}) we now conclude that
\begin{equation}\label{mbd2}
  \widetilde m_\mu(\Delta)= \frac 1{ \lambda \langle u| \dig V \dig^*|  u\rangle
 - \lambda^2\langle u|  \dig V M_\Delta V  \dig^*| u\rangle
   +O(\lambda^3)}\,,
\end{equation}
where $u(p)=(4\pi\mu)^{-1/2}$ is the normalized constant function on the sphere $\Omega_\mu$.
Moreover, with $\varphi = V \dig^* u$,
\begin{multline}\nonumber
  \langle u| \dig V M_\Delta V \dig^*| u\rangle  = \int_{\R^3} \frac 1{E(p)}
  |\hat\varphi(p)|^2 \, dp - \widetilde m_\mu(\Delta)
  \int_{\Omega_\mu} |\hat\varphi(\sqrt\mu p/|p|)|^2 \, d\omega(p) \\
= \int_{\R^3} \left(\frac 1{E(p)} \left[ |\hat\varphi(p)|^2
      -|\hat\varphi(\sqrt\mu p/|p|)|^2 \right] + \frac 1{p^2}
    |\hat\varphi(\sqrt\mu p/|p|)|^2 \right) dp \,.
\end{multline}
Using Lemma~\ref{Deltach} and the fact that $\lim_{\lambda \to 0}
f(\lambda)=0$, we conclude that
\begin{equation}
 \lim_{\lambda\to 0} \langle u|  \dig V M_\Delta V \dig^* | u\rangle  =  \langle u| \W_\mu| u\rangle \,,
\end{equation}
with $\W_\mu$ defined in (\ref{defW}). (Compare with
Eqs.~(\ref{comeq}) and~(\ref{deno2}).)  In combination with
(\ref{calcmtp}) and (\ref{mbd2}) and the definition of $\B_\mu$ in
(\ref{defB}), this proves that
$$
\lim_{\lambda \to 0} \left(\ln\left(\frac\mu {\Delta(\sqrt\mu)}
  \right) + \frac {\pi}{2 \sqrt{\mu}\, \langle u| \B_\mu|u\rangle }\right) = 2 -
\ln(8)\,.
$$
The same holds true with $\langle u|\B_\mu|u\rangle$ replaced by
$b_\mu(\lambda) = \infspec \B_\mu$, since under our assumptions
on $V$ the two quantities differ only by terms of order $\lambda^3$.

Now, by the definition of the energy gap $\Xi$ in (\ref{defxi}), $\Xi\leq
\Delta(\sqrt\mu)$. Moreover,
$$
\Xi \geq \min_{|p^2-\mu|\leq \Xi} |\Delta(p)|\,,
$$
from which it easily follows that $\Xi \geq \Delta(\sqrt\mu)( 1-
o(1))$, using Lemma~\ref{Deltach}. This proves Theorem~\ref{gap}.

\section*{Acknowledgments}
R.S. gratefully acknowledges partial support by U.S. National Science Foundation
grant PHY 0652356 and by an A.P. Sloan Fellowship.


\end{document}